\documentclass[preprint,showpacs,preprintnumbers,amsmath,amssymb,nofootinbib]{revtex4}

\usepackage{etex}
\usepackage{amssymb,amsthm,amscd,amsbsy,array}
\usepackage{bm}
\usepackage{soul} 
\usepackage{graphics,graphicx,xcolor}
\usepackage{soul} 
\usepackage{algorithm}
\usepackage{algorithmic}
\usepackage{longtable}
\usepackage{graphicx}
\usepackage{epstopdf}
\usepackage{subfigure}
\usepackage{float}
\usepackage[utf8]{inputenc}
\usepackage[T1]{fontenc}
\usepackage[english]{babel}
\usepackage{graphicx}
\usepackage{float}
\usepackage{tikz}
\usepackage{xcolor}
\usepackage[utf8]{inputenc}
\usepackage[T1]{fontenc}
\usepackage[english]{babel}
\usepackage{graphicx}
\usepackage{float}
\usepackage{tikz}
\usepackage{xcolor}
\usepackage{makecell}
\usepackage{multirow}

\usepackage[colorlinks=true, pdfstartview=FitV, linkcolor=blue, citecolor=blue, urlcolor=blue]{hyperref} 

\newcommand{\gb}{\quad\colorbox{green}}

\newenvironment{redtext}{\color{red}}
{\ignorespacesafterend}
\newenvironment{bluetext}{\color{blue}}{\ignorespacesafterend}

\newenvironment{magentatext}{\color{magenta}}{\ignorespacesafterend}
\newenvironment{cyantext}{\color{cyan}}{\ignorespacesafterend}
\newenvironment{orangetext}{\color{orange}}
{\ignorespacesafterend}

\newcommand{\bmagenta}{\begin{magentatext}}
	\newcommand{\emagenta}{\end{magentatext}}
\newcommand{\bcyan}{\begin{cyantext}}
	\newcommand{\ecyan}{\end{cyantext}}
\newcommand{\bblue}{\begin{bluetext}}
	\newcommand{\eblue}{\end{bluetext}}
\newcommand{\bred}{\begin{redtext}}
	\newcommand{\ered}{\end{redtext}}
\newcommand{\borange}{\begin{orangetext}}
	\newcommand{\eorange}{\end{orangetext}}

\numberwithin{equation}{section}

\let\ssection=\section
\renewcommand{\section}{\setcounter{equation}{0}\ssection}
\newcommand{\beq}{\begin{equation}}
	\newcommand{\eeq}{\end{equation}}

\def\s.t.{{\quad\text{\small such that}\quad}}

\newcommand{\Tr}{\mathrm{Tr}}

\def\smallover\#1/\#2{\hbox{$\textstyle\frac{\#1}{\#2}$}} %

\newtheorem{theorem}{Theorem}

\newtheorem{proposition}{Proposition}

\def\benu{\begin{enumerate}}
	\def\eenu{\end{enumerate}}
\def\bitem{\begin{itemize}}
	\def\eitem{\end{itemize}}

\def\beq{\begin{equation}}
	\def\eeq{\end{equation}}
\def\beqa{\begin{eqnarray}}
	\def\eeqa{\end{eqnarray}}

\def\barray{\left(\begin{array}}
	\def\earray{\end{array}\right)}
\def\barraynb{\begin{array}}
	\def\earraynb{\end{array}}

\def\?{{\quad\gb{\fbox{\texttt{?}}\;}}\quad}

\def\v0{\mathbf{0}}

\def\smallover#1/#2{\hbox{$\textstyle\frac{#1}{#2}$}} %
\def\smallcirc{{\raise 0.5pt \hbox{$\scriptstyle\circ$}}}
\def\cabove(#1){\stackrel{\smallcirc}{#1}}
\def\ccabove(#1){\stackrel{\smallcirc\smallcirc}{#1}}
\def\cccabove(#1){\stackrel{\,\smallcirc\smallcirc\smallcirc}{#1}\,}
\def\2{{\smallover1/2}}

\let\ssection=\section
\renewcommand{\section}
{\setcounter{equation}{0}\ssection}

\def\besub{\begin{subequations}}
	\def\esub{\end{subequations}}

\begin{document}

\preprint{arXiv: 2312.16133v3 [hep-th]}

\title{The non-Abelian Aharonov-Bohm-effect}

\author{
	P.~A. Horv\'athy
	\footnote{Originally appeared as 
	BI-TP-82/14, available in inspire (code {Horvathy:1982fx}). The author's present address is~: Institut Denis Poisson, Tours University -- Orl\'eans University, UMR 7013 (France). mailto: horvathy@univ-tours.fr}
}

\affiliation{  
	Fakult\"{a}t f\"{u}r Physik	
	Universit\"{a}t Bielefeld \\
		D-4800 Bielefeld 1, Fed.Rep. Germany \\
		and \\
		Centre de Physique Th\'{e}orique, Luminy, \\
		C.N.R.S. Marseille (France)\\
}

\begin{abstract}

The scattering of a nucleon beam around a cylinder containing a non-Abelian flux is studied. We confirm all the previsions of Wu and Yang. We consider the generalization to the gauge group $SU(N)$, and derive a classification scheme. Isospin precession is recovered also at the classical limit.

\end{abstract}

\maketitle

\tableofcontents

\section{Introduction}
The generalized Bohm-Aharonov (BA) experiment has been proposed by Wu and Yang \cite{c1} to test the existence of Yang-Mills (Y-M) fields. They suggest to scatter a nucleon beam around a thin solenoid containing a non-Abelian flux. Although the particles do not penetrate into the cylinder, there will be an interaction between the nucleons and the gauge potential. The interference pattern we observe on the screen behind the cylinder will be characteristic of this interaction.

The main conceptual as well as experimental difficulty with this non-Abelian BA effect concerns the creation of a Yang-Mills field with vanishing field strength outside of a cylinder. We are unable to clarify this point. Thus we propose instead to describe the scattering \underline{supposed} such a field has been created somehow.

In order to simplify the problem we use a number of approximations:
\begin{itemize}
	\item We freeze out the dynamical degrees of freedom of the Yang-Mills field. Hence we study the motion of a test particle in an external Y-M field, rather than solve coupled field equations.	
	\item We suppose our particles to be spinless.
	\item We describe our particles non-relativistically. Thus we use the Schr\"{o}dinger rather than the Klein-Gordon or Dirac equation.
	\item The electromagnetic interaction is not taken into account. We focus our attention on strong interactions.
\end{itemize}

Firstly we discuss all Yang-Mills fields $A_{\mu}$ with $F_{\mu v}=0$ outside of a cylinder. They are characterized by the non-integrable phase factor \cite{c1,c9} of Wu and Yang. We show that there exists a specific gauge which we call the ``diagonal gauge'' such that the $SU(N)$-vacuum splits to $N$ independent 
``electromagnetic'' BA-vacua. In this gauge the Schr\"{o}dinger equation is solved immediately by mere application of the known results from the electromagnetic case \cite{c2,c4,c5}. In particular, the $S$-matrix is constructed and this allows to prove all heuristic previsions of Wu and Yang.

Next we study the scattering of elementary particles with some generic internal structure on $SU(N)$ fluxes. Surprisingly it may happen that the outcome of the experiment is the same in gauge-inequivalent fields (cf. Table \ref{Table1}).

It is interesting to notice that a ``remnant'' of the quantum effect is found also at the classical level. In particular, Wong's equations \cite{c3} are solved explicitly and the relation of the solution to the quantum results is explained.

\section{Yang-Mills vacua in the generalized Bohm-Aharonov experiment}

Let $A_{\mu}$ be a Yang-Mills (Y-M) potential for an $SU(N)$ Y-M field whose field strength tensor vanishes outside of a thin solenoid along the $z$ axis. Dropping the $z$ variable we work on $M=R^{2} \backslash \{0\}$.

According to Wu and Yang \cite{c1} gauge fields have to be described by the non-integrable phase factor
\begin{eqnarray}
\phi(\gamma)=\exp\left(\int_{\gamma} A_{\mu} d x^{\mu}\right). \label{2.1}
\end{eqnarray}
Now we show this explicitely. Note first that the value of \eqref{2.1} is the same for homotopic paths. Let us denote the $SU(N)$ element associated to the loop which circles once around the origin in counter-clockwise direction simply by $\phi$. Under a local gauge transformation defined by the $SU(N)$-valued function $g(x)$, $\phi$ transforms according to $\phi \rightarrow g^{-1}\left(x_{0}\right) \phi g\left(x_{0}\right)$, where $x_{0}$ is the initial and endpoint of the loop $\gamma$. Conversely, \cite{c6}
\begin{theorem}
	Two Y-M vacua $A^{(1)}$ and $A^{(2)}$ are gauge-equivalent if there exists a fixed element $h \in S U(N)$ such that
	\begin{eqnarray}
	\phi_{1}=h^{-1} \phi_{2} h. \label{2.2}
	\end{eqnarray}
\end{theorem}
\begin{proof}
	If \eqref{2.2} holds, we apply first a global gauge transformation on $A^{(2)}$ defined by $h$. Then the two loop factors are equal. Let us cut the plane along the positive $x$ axis and work over $\widetilde{M}=\{(r, \theta) \in(0, \infty) \times[0,2 \pi]\}$. Choose an $\widetilde{x}_{0} \in \widetilde{M}$ arbitrarily. Any $\widetilde{x} \in {M}$ can be joined by a $\gamma$ to $\widetilde{x}_{0}$ and the associated $\phi(\gamma)$ is actually independent of $\gamma$, for $\widetilde{M}$ is already simply connected. We get thus well-defined functions on $\widetilde{M}$
	\begin{eqnarray}
		\widetilde{g}_{j}(x)=\phi_{j}(\gamma) \quad (j=1,2). 
	\end{eqnarray}
	They satisfy $\partial_{\mu} \widetilde{g}_{j}=-A_{\mu} \widetilde{g}_{j}$. The gauge transformation over $\widetilde{M}$ defined by
    \begin{eqnarray}
	\widetilde{g}(x)=\widetilde{g}_{1}(x) \widetilde{g}_{2}^{-1}(x),
    \end{eqnarray}
	carries $A^{(1)}$ to $A^{(2)}$. It descends to $M$ if an only if $\widetilde{g}(r, 2 \pi)=\widetilde{g}(r, 0)$, $r>0$. But this happens just when $\phi_{1}=\phi_{2}$.
\end{proof}

Since any element of $S U(N)$ can be diagonalized if we conjugate it with a suitable $h \in SU(N)$, there exist a gauge where $\phi$ is diagonal:

\begin{equation}
	\phi=
	\left(
	\begin{array}{ccc}
		e^{2\pi i\alpha_1} & \ & \ \\
		\ & \ddots & \ \\
		\ & \ & e^{2\pi i\alpha_N}
	\end{array}
	\right), \label{2.5}
\end{equation}
where $\alpha_{j}$ is real and $\sum\limits^{N}\alpha_{j}=0$. The factors in the diagonal are unique up to permutation.

A gauge potential for \eqref{2.5} is found at once:
\begin{eqnarray}
	A_r=0, \quad 
	A_{\theta}=\frac{1}{i}\left(
	\begin{array}{ccc}
		\alpha_1 & \ & \ \\
		\ & \ddots & \ \\
		\ & \ & \alpha_N
	\end{array}
	\right), \quad
	A_0=0.
\end{eqnarray}

We shall refer to the specific gauge where $A_{\mu}$ has this simple form as to the ``diagonal gauge''.

The $\alpha_{j}$ is in the diagonal define $N$ independent $U(1)$ fields on $M$. We can identify them with uncoupled ``electromagnetic BA fields'' with flux $2 \pi \alpha_j$, if we take the electric charge $1$.

\section{Nucleon scattering around a Y-M flux}

Let us consider an $S U(2)$ vacuum $A_{\mu}$ on $M$ characterized by
\begin{eqnarray}
\phi=\left(\begin{array}{cc}
	e^{i 2 \pi \alpha} & 0 \\
	0 & e^{-i 2 \pi \alpha}
\end{array}\right). \label{3.1}
\end{eqnarray}

We describe the motion of a nucleon in this field by
\begin{eqnarray}
-\frac{1}{2m}\left(\partial_{j}-A_{j}\right)^{2} \psi=i \frac{\partial \psi}{\partial t}\,, 
\label{3.2} 
\end{eqnarray}
where the $SU(2)$ doublet $\psi$ is the nucleon wave function. The proton and neutron states are defined as the eigenfunctions belonging to the eigenvalues $\pm \frac{1}{2}$ of the isospin operator $\widehat{J}$. The form of $\widehat{J}$ depends on the gauge we use. We call ``natural gauge'' the one whose $\widehat{J}$ reads
\begin{eqnarray}
\widehat{J}=\frac{1}{2} \sigma_{3}=\frac{1}{2}\left(\begin{array}{cc}
	1 & 0 \\
	0 & -1
\end{array}\right) \,.
\label{3.3}
\end{eqnarray}
Under a local gauge transformation $\widehat{J}$ goes to ${g} \widehat{J} {g}^{-1}$.

Let us suppose that \eqref{3.2} stands in the natural gauge. In order to solve it, we go over to the diagonal gauge. Indeed, if the appropriate gauge transformation belongs to $g$, then $\Psi=g \psi$ satisfies
\begin{eqnarray}
\widehat{H} \Psi=\left(\begin{array}{cc}
	\widehat{H}_{\alpha} & 0 \\
	0 & \widehat{H}_{-\alpha}
\end{array}\right)\left(\begin{array}{l}
	\Psi_{1} \\
	\Psi_{2}
\end{array}\right)=i \frac{\partial}{\partial t}\left(\begin{array}{l}
	\Psi_{1} \\
	\Psi_{2}
\end{array}\right)\,,
\end{eqnarray}
where $\hat{H}_{\alpha}$ denotes
\begin{eqnarray}
\widehat{H}_{\alpha}=-\frac{1}{2m}\left(\frac{\partial^{2}}{\partial r^{2}}+\frac{1}{r} \frac{\partial}{\partial r}+\frac{1}{r^{2}}\Big(\frac{\partial}{\partial \theta}+i \alpha\Big)^{2}\right)\,. 
\label{3.5}
\end{eqnarray}

Here we recognize the Hamiltonian of a \emph{particle with charge 1 in an electromagnetic BA field \cite{c2} with flux $\phi=(-2 \pi a) / \mathrm{e} \, $}. \eqref{3.5} is thus solved by mere application of the results known from the electromagnetic case. This allows us to \underline{prove} the heuristic statements of Wu and Yang as follows:
\begin{itemize}
\item The solution of the Schr\"{o}dinger equation $\widehat{H}_{\alpha} \psi=i \partial \psi / \partial t$ is known to depend on $e^{i2\pi\alpha}$. But knowing $e^{i2\pi \alpha}$ is the same as knowing $\phi$ in \eqref{3.1}.
\end{itemize}

We conclude
\begin{proposition}
The outcome of the generalized BA experiment depends only on the Wu-Yang phase factor $\phi$.
\end{proposition}

The standard notions of scattering theory are introduced in the usual way \cite{c5,c10}. In the diagonal gauge the S-matrix reads obviously
\begin{eqnarray}
\widehat{S}=\left(\begin{array}{cc}
	\widehat{S}_{\alpha} & 0 \\
	0 & \widehat{S}_{-\alpha}
\end{array}\right),
\end{eqnarray}
where $\widehat{S}_{\alpha}$ is the S-matrix for the electromagnetic BA effect \footnote{Note added: $\alpha$ in the Abelian case is indeed the half of the non-Abelian one.}.

In order to derive our further results, let us summarize briefly what is known in the Abelian case \cite{c4,c5}.

We recall that in the gauge $A_{0}=0$, $A_{r}=0$, $A_{\theta}=-\alpha$ the angular momentum operator in the electromagnetic BA field reads
\begin{eqnarray}
\hat{I}=-i \frac{\partial}{\partial \theta}+\alpha\,.
\end{eqnarray}

Its spectrum is given by
\begin{eqnarray}
\lambda_{m}=m+\alpha, \quad m=0, \ \pm1, ...
\end{eqnarray}

The momentum Hilbert space $H=L^{2}\left(R^{2}, d \vec{k}\right)$ splits according to $H=\bigoplus\limits_{m=-\infty}^{\infty} H_m$ where the angular momentum subspace 
$H_m$ is defined by the projection map
\begin{eqnarray}
\left(\widehat{P}_{m} \psi\right)(k, \theta)=\frac{1}{2 \pi} e^{i m \theta} \int_{\pi}^{\pi} d \varphi e^{-i m \varphi} \psi(k, \varphi)\,.
\end{eqnarray}

In momentum representation the $S$-matrix has a particularly simple form
\begin{eqnarray}
\left.\hat{\mathbf{S}}\right|_{H_{m}}=e^{2 i \delta_{m}(\alpha)}
\end{eqnarray}
where
\begin{eqnarray}
2 \delta_{m}(\alpha)= \begin{cases}-\pi \alpha & \text { for } \lambda_{m} \geq 0 \\
\;\;  \pi \alpha & \text { for } \lambda_{m}<0\end{cases} 
\label{3.11}
\end{eqnarray}

It is the sign of $\lambda_{m}$ which determines the phase shift. \eqref{3.11} is understood intuitively by noting that the sign of $\widehat{I}$ depends on which side the particle passed by the origin. The BA effect is thus the consequence of the fact that quantum particles may be split and go simultaneously in both sides.

In the time-independent approach one works with incident plane wave $\psi_{\mathrm{inc}}=e^{i \vec{k}\cdot\vec{x}}$. The outgoing wave has the form
\begin{eqnarray}
\psi_{\mathrm{out}}=e^{i \vec{k}\cdot\vec{x}}+f(k, \theta) \cdot \frac{e^{i k r}}{r^{\frac{1}{2}}}\,,
 \nonumber
\end{eqnarray}
where
\begin{eqnarray}
f(k, \theta)=\left(\frac{1}{2 \pi k i}\right)^{1 / 2} \frac{\sin (\pi \alpha)}{\sin(\theta/2)} e^{-i\left(\frac{1}{2}+[\alpha]\right) \theta}, \quad \theta\neq0. \label{3.12}
\end{eqnarray}

Now let us return to the non-Abelian case. Suppose the incident wave is given in momentum representation. Let us split as
\begin{equation}
	\varphi_{\mathrm{inc}}=\varphi_{+}+\varphi_{-}, \quad \chi_{\mathrm{inc}}=\chi_{+}+\chi_{-}, \nonumber
\end{equation}
where $\varphi_{+},\, \chi_{+}$ and $\varphi_{-}, \,\chi_{-}$ are the positive and negative angular momentum parts. The outgoing wave function reads
\begin{eqnarray}
\psi_{\text {out }}=\left(\begin{array}{l}
	\varphi_{\text {out }} \\[4pt]
	\chi_{\text {out }}
\end{array}\right)=\widehat{S} \psi_{\mathrm{inc}}=\left(\begin{array}{c}
	e^{-i \pi \alpha}{\varphi_{+}}+e^{i \pi \alpha}{\varphi_{-}} \\[4pt]
	e^{i \pi \alpha} \chi_{+}+e^{-i \pi \alpha} \chi_{-}
\end{array}\right)\,. 
\label{3.13}
\end{eqnarray}
All statements of \cite{c1} now follow at once. Let us first suppose that the direction of $A_{\mu}$ is the same as the proton-neutron direction in isospin space. In other words the natural and the diagonal gauges are the same.

From \eqref{3.13} we conclude
\begin{proposition}
The incoming protons $(\chi=0)$ and neutrons $(\varphi=0)$ get opposite phase shifts.
\end{proposition}

Let the incident beam now be a coherent mixture of neutrons and protons in a pure state. For the sake of simplicity let us suppose that
\begin{eqnarray}
\varphi_{\mathrm{i n c}}=\chi_{\mathrm{i n c}}=\frac{1}{\sqrt{2}} e^{i\vec{k}\cdot\vec{x}}\,. 
\nonumber
\end{eqnarray}

Using \eqref{3.12} we get
\begin{proposition}
	The outgoing nucleon intensity, $\iota=|\varphi|^{2}+|\chi|^{2}$ and the outgoing mixing ratio $\mu=|\varphi| /|\chi|$ \underline{fluctuate} according to
	\begin{eqnarray}
		&&\iota=1+\frac{1}{2 \pi k r} \frac{\sin ^{2}(\pi a)}{\sin ^{2}(\theta / 2)}+\frac{\sin (\pi a)}{\sin (\theta/2)}\left(\frac{1}{2 \pi k r^{2}}\right)^{1 / 2} 
\\[12pt]
		&&\qquad \left[\cos \left(\frac{\pi}{4}+\left(\frac{1}{2}+[\alpha]\right) \theta-k r(1+\cos \theta)\right)-\cos \left(\frac{\pi}{4}+\left(\frac{1}{2}+[-\alpha]\right) \theta-k r(1+\cos \theta)\right)\right], \nonumber \\ [14pt]
		&&\mu=  \\ [12pt]
		&&\left[ \frac{(2 \pi k r)\sin^{2}\frac{\theta}{2}+\sin^{2}\pi\alpha+(2\pi kr)^{\frac{1}{2}}2\sin(\pi\alpha)\sin\left(\frac{\theta}{2}\right)\cos\left[\frac{\pi}{4}+\left(\frac{1}{2}+[\alpha]\right)\theta-kr(1+\cos\theta)\right]}
		{(2 \pi k r) \sin^{2}\frac{\theta}{2}+\sin^{2}\pi\alpha-(2\pi kr)^{\frac{1}{2}}2\sin(\pi\alpha)\sin\left(\frac{\theta}{2}\right)\cos\left[\frac{\pi}{4}+\left(\frac{1}{2}+[-\alpha]\right)\theta-kr(1+\cos\theta)\right] }\right]^{\frac{1}{2}}. \nonumber 
	\end{eqnarray}
	
\end{proposition}

Let us now imagine that the field's direction differs from the proton direction. When bringing $\left(A_{\mu}\right)$ to the diagonal form $\widehat{J}$ looses its initial form \eqref{3.3}. Suppose, for example, that it becomes $\frac{1}{2} \sigma_{1}$. Protons and neutrons are represented in this case by
\begin{eqnarray}
\psi_{p}=\left(\begin{array}{l}
	\varphi \\
	\varphi
\end{array}\right) \quad \mathrm{and} \quad \psi_{n}=\left(\begin{array}{r}
	\chi \\
	-\chi
\end{array}\right). \nonumber
\end{eqnarray}

Let us scatter now a pure proton or pure neutron beam. From \eqref{3.13} we derive: 
\begin{proposition}
Some of the incoming protons (neutrons) become neutrons (protons) after scattering.
\end{proposition}

In particular, if all particles go in one side - we mean by that $\varphi_{-}$ or $\varphi_{+}=0$  $\left(x_{-} \ \mathrm{or} \ x_{+}=0\right)$ then in the diagonal gauge the $S$-matrix is just an $SU(2)$-rotation by angle $\mp 2 \pi \alpha$ :
\begin{eqnarray}
\widehat{S}_{ \pm}=\widehat{R}_{\mp 2 \pi \alpha}, \label{3.16}
\end{eqnarray}
where plus or minus in the left hand side refers to the sign of the angular momentum. We call this effect \underline{isospin precession}.

Remark that $\widehat{S}_{ \pm}$ in \eqref{3.16} is just the square root of the Wu-Yang factor $\phi$.
\begin{eqnarray}
\widehat{S}_{ \pm}=(\phi)^{\mp \frac{1}{2}}=\left(\begin{array}{ll}
	e^{\mp i \pi \alpha} & \\
	& e^{\pm i \pi \alpha}
\end{array}\right) \,.
\label{3.17}
\end{eqnarray}
Equation 
\eqref{3.17} is explained intuitively as follows: if our particle goes entirely on one side, its polar angle changes by $\pm \pi$. So we get the half of the shift associated to a change of $2 \pi$ -- the one corresponding to $\phi$.

If $\alpha=\frac{1}{2}$, when the electromagnetic cross section is maximal \cite{c2,c4,c5}, the S-matrix reads
\begin{eqnarray}
	&& \widehat{S}\left(\begin{array}{l}
		\varphi \\
		\varphi
	\end{array}\right)=-i\left(\begin{array}{r}
		\chi \\
		-\chi
	\end{array}\right), \quad \mathrm{ where } \quad \chi=\varphi_{+}-\varphi_{-}, 
\\[12pt]
	&& \widehat{S}\left(\begin{array}{r}
		\chi \\
		-\chi
	\end{array}\right)=-i\left(\begin{array}{l}
		\varphi \\
		\varphi
	\end{array}\right), \quad \mathrm{ where } \quad \varphi=\chi_{+}-\chi_{-}\,.
\end{eqnarray}

In this case \emph{all protons are turned into neutrons and all neutrons are turned into protons}~!

\section{Generalization to $SU(N)$ \label{sec4}}

The same experiment can be imagined also with other particles. For example, $\pi$ mesons admit an $SU(2)$ internal structure and can thus be coupled to an $SU(2)$-field. Quarks admit an $SU(3)$-structure and can be scattered on a gluon field, etc. We study here the scattering of a particle belonging to an Unitary Irreducible Representation (UIR) of $SU(N)$.

Denote $G$ the Lie algebra of the gauge group $G=SU(N)$. The generalized isospin operator $\widehat{J}$ consists now of a maximal commutative subalgebra $H$ of G. For $G=S U(N)$ $H$ has $(N-1)$ generators $\underline{H}=\left(H_{1}, \ldots H_{N-1}\right)$. When applying a gauge transformation $\underline{H}$ changes according to $\underline{H} \rightarrow g\underline{H}g^{-1}=\left(gH_{1}g^{-1}, \cdots, gH_{N-1} g^{-1} \right)$. We can thus rotate $H$ by a gauge transformation to the diagonal. The special gauge where all the $\widehat{U}\left(H_{i}\right)$'s are diagonal matrices we call again natural gauge. For example, in case of $SU(3)$ the usual choice is \cite{c11}
\begin{eqnarray}
H_{1}=\frac{1}{\sqrt{6}}\left(\begin{array}{ccc}
	1 & \ & \ \\
	\ & -1 & \ \\
	\ & \ & 0
\end{array}\right), \quad H_{2}=\frac{1}{3 \sqrt{2}}\left(\begin{array}{ccc}
1 & \ & \ \\
\ & 1 & \ \\
\ & \ & -2
\end{array}\right)
\end{eqnarray}
$H_{1}$ here is the isospin, $H_{2}$ the hypercharge operator.

An UIR is characterized by $(N-1)$ non-negative integers $P_{1} \geq \ldots P_{N-1} \geq 0$. $P_{j}$ is the length of the $\widehat{J}$ 's row in the Young tableau \cite{c11}.

Denote $\widehat{X}_{k}$ ($k=1,...,\mathrm{dim} \ SU(N)$) the generators of the UIR $\widehat{U}$. Let the dimension of the representation be $n$ $(\geq N)$.

We describe the scattering by
\begin{eqnarray}
-\frac{1}{2m}\left(\partial_{j}-A_{j}^{k} \widehat{X}_{k}\right)^{2} \psi=i \frac{\partial \psi}{\partial t} \label{4.2}
\end{eqnarray}
where $\psi$ is now an $n$-tuplet. \eqref{4.2} is again in natural gauge. Just like in the preceding $\S$, we solve it by switching to diagonal gauge
\begin{eqnarray}
	&& A_{\mu} \rightarrow g^{-1} A_{\mu} g+g^{-1} \partial_{\mu} g \nonumber 
	\\[4pt]
	&& \psi \rightarrow \widehat{U}(g) \psi=\Psi 
	\\[4pt]
	&& \widehat{U}(\underline{H}) \rightarrow \widehat{U}(g) \ \widehat{U}(\underline{H}) \  \widehat{U}\left(g^{-1}\right) \nonumber
\end{eqnarray}

We get thus $\underline{n}$ uncoupled scalar equations for the components of $\Psi$
\begin{eqnarray}
\widehat{H}_{a j} \Psi_{j}=i \frac{\partial\Psi_{j}}{\partial t}, \quad  j=1...n\,, 
\label{4.4}
\end{eqnarray}
where the new ``fluxes'' are defined by
\begin{eqnarray}
\widehat{U}(\phi)=\left(\begin{array}{ccc}
	e^{2 \pi i \alpha_{1}} &  & \\
	& \ddots & \\
	& & e^{2 \pi i \alpha_{n}}
\end{array}\right)\,.
\end{eqnarray}
Equation \eqref{4.4} is solved again by the electromagnetic solutions. As these solutions are characterized by the $e^{2 \pi i a_{j}}$'s $(j=1...n)$, we conclude
\begin{proposition}
	The outcome of the generalized BA experiment with particles belonging to the UIR $\widehat{U}$ depends on $\widehat{U}(\phi)$.
\end{proposition}
 
Similarly the S-matrix reads now
\begin{eqnarray}
\widehat{S}=\left(\begin{array}{ccc}
	\widehat{S}_{\alpha_{1}} & & \\
	 & \ddots & \\
	& & \widehat{S}_{\alpha_{n}}
\end{array}\right)=\widehat{U}\left(S^{0}\right),
\end{eqnarray}
where $S^{0}$ is the S-matrix for particles belonging to the fundamental representation. If the particle travels on one side, then $\widehat{S}$ is again
\begin{eqnarray}
	\widehat{S}^{\pm}=\left(\widehat{U}(\phi)\right)^{\mp\frac{1}{2}}
	=\widehat{U}\left(\phi^{\mp\frac{1}{2}}\right)=
		\left(
	\begin{array}{ccc}
		e^{\mp i\pi \alpha_1} & \ & \ \\
		\ & \ddots & \ \\
		\ & \ & e^{\mp i\pi \alpha_n}
	\end{array}
	\right),	\label{4.7}
\end{eqnarray}
where plus and minus refer to the sign of the angular momentum.

The surprising consequence is that if $\widehat{U}$ is \underline{not faithful}, we may have the \underline{same scattering} also in \underline{different fields}! (This has been noted independently by Asorey \cite{c14} as we learnt from a preprint received after this work has been completed).

The question for which fields does it happen is answered by using some representation theory \cite{c7}.

\begin{theorem}
Let $\phi_{1}=\operatorname{diag}\left(\mathrm{e}^{2 \pi i \alpha_{j}}\right)$ and $\phi_{2}=\operatorname{diag}\left(e^{2 \pi i \beta_{j}}\right)$ characterize two $SU(N)$ vacua over M. $\widehat{U}\left(\phi_{1}\right)=\widehat{U}\left(\phi_{2}\right)$ and the outcome of the generalized $B A$ experiment is thus the same if and only if
\begin{eqnarray}
&&\phi_{1}=\varepsilon \phi_{2}, \quad \mathrm{ where } \quad \varepsilon^{N}=1\,, 
\\[8pt]
&&{\varepsilon}^{\sum\limits^{N-1}P_j}=1\,. 
\label{4.9}
\end{eqnarray}
\end{theorem}

This reads also
\begin{eqnarray}
e^{2 \pi i \alpha_{j}}=e^{\frac{2 \pi i}{N} \ell} e^{2 \pi i \beta_{j}}, \quad \forall j=1...N, \quad 0 \leq\ell<N,
\end{eqnarray}
and
\begin{eqnarray}
\prod^{N-1}\left(e^{2 \pi i \alpha_{j}}\right)^{P_j}=\prod^{N-1}\left(e^{2 \pi i \beta_{j}}\right)^{P_j}. 
\end{eqnarray}

The physically most interesting cases are given in Table \ref{Table1}. Note that in the $U(1)$ case no analogue of (\ref{4.9}) exists.

As an illustration let us study the scattering of \underline{pions}. The $\pi^{+}, \pi^{0}$ and $\pi^{-}$ mesons correspond to representation $\{3\}$ of $SU(2)$ \cite{c12}. Their wave functions are the eigenvectors of the 3-dimensional isospin-operator
\begin{eqnarray}
\widehat{U}(J)=\frac{1}{2}\left(\begin{array}{ccc}
	1 & & \\
	& 0 & \\
	& & -1
\end{array}\right),
\end{eqnarray}
with eigenvalues $1$, $0$ and $-1$, eigenfunctions
\begin{eqnarray}
\psi_{\pi^{+}}=\left(\begin{array}{l}
	\varphi \\
	0 \\
	0
\end{array}\right), \quad \psi_{\pi^{0}}=\left(\begin{array}{l}
	0 \\
	\chi \\
	0
\end{array}\right), \quad \psi_{n^{-}}=\left(\begin{array}{l}
	0 \\
	0 \\
	\eta
\end{array}\right).
\end{eqnarray}
$\widehat{U}(\phi)$ reads now
\begin{equation}
	\widehat{U}(\phi)=\left(\begin{array}{ccc}
		e^{4 \pi i \alpha} & \\
		& 1 & \\
		& & e^{-4 \pi i \alpha}
	\end{array}\right).
\end{equation}

We see that if the diagonal and the natural gauges are the same, then the $\pi^{+}$ or the $\pi^{-}$ mesons get opposite phase shifts while $\pi^{0}$ particles pass undisturbed.

$\widehat{U}(\phi)$ is the same for $\alpha$ and $\alpha+\frac{1}{2}$. Pions scatter hence identically, although the Y-M fields are not gauge-equivalent.

\section{Classical limit}

We would like to end this paper with some observations concerning the classical limit. Classical particles with internal structure have been introduced by Wong \cite{c3} (see also \cite{c13}). They move according to
\begin{eqnarray}
	&& m\ddot{x}_{\mu}=\Tr\left(J F_{\mu\nu}\dot{x}^{\nu}\right), \\
	&& \dot{J}=\left[J, A_{\mu} \dot{x}^{\mu}\right]. \label{5.2}
\end{eqnarray}

The internal ``classical isospin'' variable ${J}$ here belongs to the Lie algebra $G$ of the gauge group.

$F_{\mu\nu}=0$ implies that the particles travel along straight lines with constant speed. Their isospin is, however, parallel transported. The general solution of \eqref{5.2} along an arbitrary path $\gamma(t) \subset M$ reads
\begin{eqnarray}
J(t)=A d_{\phi[\gamma(t)]}J_{0}=\phi_{\gamma(t)} J_{0} \phi_{\gamma(t)}^{-1}. \label{5.3}
\end{eqnarray}
\eqref{5.3} implies the particle stays on a fixed adjoint orbit ${O}_{J_{0}}=\left\{{Ad}_{G} \ {J}_{0}\right\}$ \cite{c3,c13}.

In the diagonal gauge \eqref{5.3} has a particularly simple form. Its matrix ${J}_{j k}$ reads
\begin{eqnarray}
J_{j k}(t)=e^{i\left(\alpha_{k}-\alpha_{j}\right)[\theta(t)-\theta(0)]} J_{j k}(0), \label{5.4}
\end{eqnarray}
where $\theta$ is the polar angle. Remarkably $J$ depends on the position rather than on time.

Since the space-time motion is just the free one, these degrees of freedom can be dropped. The classical S-matrix becomes thus an $SU(N)$-rotation acting on the orbit $G_{J_{0}}$. It carries an ``incident isospin'' $J_{\mathrm{inc}}$ to
\begin{equation}
	J_{\text {out }}={Ad}_{S}{J}_{\mathrm{inc}}=S{J}_{\mathrm{inc}}{S}^{-1}.
\end{equation}

The $SU(N)$-element is found in diagonal gauge by inserting $\Delta \theta=\varphi(\infty)-\theta(-\infty)$ to \eqref{5.4}. We observe that $\Delta \theta$ is $-\pi$ or $+\pi$ depending on what side of the origin the particle passed by. This is the same as saying if the angular momentum $\vec{I}=m \vec{v} r_{0}$ (where $\vec{v}$ is the particle's speed and $r_{0}$ is the impact parameter) is negative or positive. We conclude
\begin{eqnarray}
J_{\text {out }}=
\begin{cases}
	Ad_{\mathrm{diag\left(e^{-i \pi \alpha_{j}}\right)}}\left(J_{\mathrm{inc}}\right), \quad I \geq 0 
	\\[4pt]
	Ad_{\mathrm{diag}\left(e^{i \pi \alpha_{j}}\right)}\left({J_{\mathrm{inc}}}\right), \quad I<0
\end{cases}.  \label{5.6}
\end{eqnarray}
Here we recognize again $S^{0}$, the square root of the Wu-Yang factor $\phi$. Hence \eqref{5.6}
\begin{eqnarray}
J_{\mathrm{out}}=S^{0} J_{\mathrm{inc}}\left(S^{0}\right)^{-1}.
\end{eqnarray}
\eqref{5.6} can be compared to the quantum expression \eqref{4.7} of unsplit particles.

We illustrate these results on the example of $SU(2)$. Its Lie algebra can be identified to $\mathbb{R}^{3}$; the orbits are 2-spheres. \eqref{5.2} reads now
\begin{eqnarray}
\dot{\vec{J}}=2 \vec{J} \times \vec{A}_{\theta} \,\dot{\theta}, \label{5.8}
\end{eqnarray}
where
\begin{eqnarray}
	\vec{A}_{\theta}=
	\left(
	\begin{array}{r}
		0 \\
		0 \\
		-\alpha
	\end{array}
	\right)\,.
\end{eqnarray}

The equation 
\eqref{5.8} is solved at once. The trajectories are horizontal circles: isospin undergoes a rotation by angle $-2\alpha \Delta \theta$ around the direction of the field. In the scattering process $\Delta \theta= \pm \pi$ depending on which side the particle went by. The classical $S$-matrix reads
\begin{eqnarray}
S_{\pm}=R_{\mp2\pi\alpha}\,.
\end{eqnarray}
We recover isospin precession \eqref{3.16} at the classical level!

\section{Discussion}

The basic difference between the electromagnetic and the non-Abelian situations in the observability of the phase. In the electromagnetic case this is a $U(1)$ factor which is unobservable. In the non-Abelian case however protons and neutrons differ just by a gauge rotation. \underline{Phase} is thus \underline{observable}. This explains how we get an observable effect even if our particle passes entirely on one side of the flux, without being split. Remarkably this is the part of the effect which does have a classical counterpart, where the particles cannot be split.

A last remark concerns the classical model of Wong \cite{c3,c13}. The point is that \eqref{5.2} is gauge-invariant only if $J$ changes according to ${J}\rightarrow {g}^{-1}J{g}$ under a gauge transformation. The classical isospin-trajectory is thus gauge-dependent and hence physically meaningless. This difficulty can be overcome in a sophisticated mathematical way \cite{c8} - but the intuitive appeal is lost.

Finally note that the classification scheme of Sec. \ref{sec4} holds also in the \underline{relativistic case} \cite{c7}.

\section{Acknowledgement}

1 would like to thank all those who helped me in preparing this paper. In particular, I am indebted to T. T. Wu for calling my attention at this problem, to J. Kollar for discussions, to M. Asorey for sending the his paper \cite{c14} before publication and to D. Miller for reading and correcting the manuscript. Hospitality and discussions with J. Potthoff and W. Pla{\ss} are acknowledged to the Theor. Phys. Dept. of Bielefeld University.

\bigskip\noindent\underline{Note added in 2023}. {\it This  paper was written in 1982 as  BI-TP-82/14 (available in inspire \cite{c15}) 
by following the recommendations of Tai-Tsun Wu during a visit at the CPT in Marseille. It has long remained unpublished because of a deep and relevant question of the referee of Phys. Rev. D concerning electric charge conservation. After a series of further attempts (which mostly remained unpublished), a thorougly revised and substantially different version appeared in 1986 \cite{c16}.  Yet another considerably different version \cite{c17} focusing at the classical isospin precession will be published to celebrate the 90th birthday of Professor Tai-Tsun Wu. The author is grateful to ``Tai'' for his advices  in the early eighties of the last century, and would also like to thank Q-L Zhao for his assistence.
}

\appendix
\section{Appendix}

\begin{table}[h]
	\centering
\begin{tabular}{|c|c|c|c|c|}
	\hline \begin{tabular}{l} 
		gauge \\
		group
	\end{tabular} & $P_{j}$ & \begin{tabular}{l} 
		represen- \\
		tation
	\end{tabular} & \begin{tabular}{l} 
		physical \\
		intexpretation
	\end{tabular} & \begin{tabular}{l} 
		number of flelds \\
		with the same \\
		solution
	\end{tabular} \\
	\hline \multirow{2}{*}{$\mathrm{SU}(2)$} & 1 & $\{2\}$ & nucleon & 1 \\ \cline{2-5}
	~ & 2 & $\{3\}$ & $\pi$-meson & 2 \\
	\hline \multirow{6}{*}{$\mathrm{SU}(3)$} & $1 \quad 0$ & $\{3\}$ & quark & \multirow{2}{*}{1} \\ \cline{2-4}
	~ & $ \ 1 \quad 1 \ $ & $\{\overline{3}\}$ & antiquark & \\ \cline{2-5}
	~ & $ \ 3 \quad 0 \ $ & $\{10\}$ & decimet & \multirow{2}{*}{3} \\ \cline{2-4}
	~ & $ \ 3 \quad 3 \ $ & $\{\overline{10}\}$ & antldecimet &  \\ \cline{2-5}
	~ & $ \ 2 \quad 1 \ $ & $\{8\}$ & octet & 3 \\ \cline{2-5}
	~ & $ \ 4 \quad 2 \ $ & $\{27\}$ & & 3 \\ \cline{2-5}
	\hline	
\end{tabular}
\caption{Scattering of particles with $SU(2)$ and $SU(3)$ internal structure on Yang-Mills fluxes. If the representation is not faithful, there exist several fields where the scattering is the same.}
\label{Table1}
\end{table}

\end{document}